\documentclass{article}
\usepackage{amsfonts}
\usepackage{amssymb}
\usepackage{graphicx,xcolor}
\usepackage{amsmath}
\usepackage{amsthm}
\usepackage{enumitem}
\usepackage{bbm}

\nonstopmode
\newtheorem{thm}{Theorem}[section]
\newtheorem{assmp}[thm]{Assumption}

\newtheorem{exmp}[thm]{Example}
\newtheorem{lem}[thm]{Lemma}
\newtheorem{notation}[thm]{Notation}

\newtheorem{prop}[thm]{Proposition}
\theoremstyle{remark}
\newtheorem{rem}[thm]{Remark}

\begin{document}

\title{Robust market-adjusted systemic risk measures}
\author{Matteo Burzoni, Marco Frittelli, Federico Zorzi\thanks{%
email: matteo.burzoni@unimi.it, marco.frittelli@unimi.it,
federicozorzi@outlook.it}}
\maketitle

\begin{abstract}
In this note we consider a system of financial institutions and study
systemic risk measures in the presence of a financial market and in a robust
setting, namely, where no reference probability is assigned. We obtain a
dual representation for convex robust systemic risk measures adjusted to the
financial market and show its relation to some appropriate no-arbitrage
conditions.
\end{abstract}


\section{Introduction}

In a system composed of $N$ financial institutions, a traditional approach
to evaluate the risk of each institution $j\in \{1,\ldots,N\}$ is to apply a 
\textit{univariate monetary risk measure} $\eta ^{j}$ to the single
financial position $X^{j}$. Once the risk $\eta ^{j}(X^{j})$ of each
institution has been determined, a naive assessment of the risk of the
entire system $X=(X^{1},\dots ,X^{N})$ could be given as the sum of the
individual risks. However, such a procedure would probably not capture the
risk of complex systems and the urge for more satisfactory measures of
systemic risk originated, in the recent years, a vast literature. \cite{ChenIyengarMoallemi} and \cite{Kromer} studied under which conditions a
systemic risk measure $\rho$ could be written in the form 
\begin{equation}
\rho (X)=\eta (\Lambda (X))=\inf \{m\in \mathbb{R}\mid \Lambda (X)+m\in 
\mathcal{A}\},  \label{DynRM2}
\end{equation}%
for some univariate monetary risk measure $\eta $ with acceptance set $%
\mathcal{A}$ and some aggregation rule $\Lambda :\mathbb{R}^{N}\rightarrow 
\mathbb{R}$ that transforms the $N$-dimensional risk factors into a
univariate risk factor. In this case, $\rho (X)$ is the minimal cash amount
that secures the system when it is added to the total aggregated loss $%
\Lambda (X)$. Note that in (\ref{DynRM2}) such a minimal capital is added 
\textit{after aggregating individual risks.} An alternative approach, see 
\cite{ACDP,systemic1,Feinstein}, proposes to add capital into the single
institutions \textit{before aggregating their individual risks} leading to
risk measures of the form: 
\begin{equation}
\rho (X):=\inf \left\{ \sum_{j=1}^{N}m^{j}\mid m=[m^{1},\cdots ,m^{N}]\in 
\mathbb{R}^{N},\,\Lambda (X+m)\in \mathcal{A}\right\} .  \label{DynRM3}
\end{equation}%
As one can see from \eqref{DynRM3}, the difference to \eqref{DynRM2} is that
each $m^{j}\in \mathbb{R}$ is added to the financial position $X^{j}$ of
institution $j\in \{1,\cdots ,N\}$ before the corresponding aggregated loss $%
\Lambda (X+m)$ is calculated. We refer the interested reader to \cite%
{systemic1} for more references on systemic risk measures. We point out that
in the above literature $X=(X^{1},\ldots ,X^{N})$ is a vector of random
variables defined on a probability space $(\Omega ,\mathcal{F},\mathbb{P})$
and consequently the acceptance set $\mathcal{A}$ is a subset of $%
L^{0}(\Omega ,\mathcal{F},\mathbb{P}).$

In this paper we depart from this literature in two respects. First, the
agents are allowed to operate in a financial market composed of $J+1$ assets 
$S^{0},S^{1},\dots ,S^{\text{J}}$ and finitely many trading periods $%
t=0,\dots ,T-1$. We make use of an abstract set $\mathcal{G}$ to describe
all possible positions that are achievable by self-financing trading
strategies with zero initial cost. Second, no assumptions are made on the
probabilities of future events. The sample space is a non-empty subset $%
\Omega \subseteq ( (0,+\infty )\times \mathbb{R}^{J}) ^{T} $, endowed with
the usual Euclidean metric. This approach is robust in the sense that we do
not impose a priori any statistical/historical probability measure on $%
\Omega $ but we rather work in a pointwise manner.

The financial position of the $N\in \mathbb{N}$ agents or financial
institutions is represented by $X=(X^{1},\dots ,X^{N})\in \mathcal{B}$,
where $\mathcal{B}:=\mathcal{B}(\mathbb{R}^{N})$ is the set of all Borel
measurable functions $\Omega \rightarrow \mathbb{R}^{N}$. Note that we also
assume that $\mathcal{G}$ is contained in $\mathcal{B}$, namely $\mathcal{G}$
is a set of \emph{vectors}. This allows us to model the case where the
agents cannot achieve the same class of terminal payoffs or the case where
they even trade in different markets. We are interested in the risk of the
entire system that we evaluate in terms of an aggregate univariate position.
To achieve this we consider an acceptance set $\mathcal{A}\subseteq \mathcal{%
B}(\mathbb{R})$ and an aggregation function $\Lambda :\mathbb{R}%
^{N}\rightarrow \mathbb{R}$. We then evaluate the risk of the financial
system by means of the following functional $\rho :\mathcal{B}\rightarrow
\lbrack -\infty ,+\infty ]$%
\begin{equation}
\rho (X):=\inf \left\{ \sum_{i=1}^{N}m^{i}\mid m\in \mathbb{R}^{N},\ \exists
g\in \mathcal{G}:\Lambda (m+X+g)\in \mathcal{A}\right\} ,  \label{def:sysrho}
\end{equation}%
with $\rho (X)=\infty $ if the set on the right hand side is empty. This
risk measure \eqref{def:sysrho} is \emph{market-adjusted}, meaning that
every agent is allowed to trade in the underlying market, in a
self-financing way and according to achievable payoffs, in order to obtain
an acceptable aggregate terminal position. We observe that this measure of
risk, that we label of the type \emph{first allocate and adjust, then
aggregate}, is in the same spirit of the risk measures (\ref{DynRM3}), even
though it has the additional market adjustment feature and is specified in a
robust framework.

\textit{Our aim is to prove a dual representation for the systemic risk
measure \eqref{def:sysrho} and to understand its interplay with possible
notions of arbitrage (see Theorem \ref{main:thm:2}).} 

To develop this theory, we will follow the same approach that \cite%
{Che:Kup:Tan} adopted for the analysis of the robust pricing-hedging duality
in one dimension and extend it to the present multivariate (systemic)
setting. We address this problem and make precise statements in Section \ref%
{sec:main}. We refer the interested reader to \cite{Che:Kup:Tan} for more
references on robustness in a non systemic framework.

An alternative way of measuring market-adjusted systemic risk employs the
use of a second aggregation function $\Gamma :\mathbb{R}^{N}\rightarrow 
\mathbb{R}$ for the payoffs of the trading strategies, i.e., by means of the
following functional $\rho_\Gamma :\mathcal{B}\rightarrow \lbrack -\infty
,+\infty ] $ 
\begin{equation}
\rho_\Gamma (X):=\inf \left\{ \sum_{i=1}^{N}m^{i}\mid m\in \mathbb{R}^{N},\
\exists g\in \mathcal{G}:\Lambda (m+X)+\Gamma (g)\in \mathcal{A}\right\} .
\label{def:sysrho2}
\end{equation}%
The interpretation is similar to the one above but it is different in
spirit. In \eqref{def:sysrho} the agents are operating as $N$ different
units both in terms of the financial position $X$ and market payoff $g$. In %
\eqref{def:sysrho2} the acceptability of the aggregate position $\Lambda
(m+X)$ can be influenced by an aggregate market payoff $\Gamma (g)$. We can
think of $\eqref{def:sysrho2}$ as the risk metric of a single firm composed
of $N$ different units and a trading desk operating independently of the $N$
units\footnote{%
For ease of notation we continue to assume that $g$ and $X$ have the same
dimension but, in principle, they could now be different.}. The total risk
of the firm is assessed by aggregating the static financial position of the
firm and the market position separately.

From a mathematical point of view there is very little difference in
treating the two cases and we present analogous results in Section~\ref%
{sec:Gamma}.

\section{The Main Results}

\label{sec:main} 

\begin{notation}
We let $\boldsymbol{1}:=(1,\dots ,1)$ be the $N$-dimensional vector with
entries all equal to $1$, so that, if $x$ is a univariate variable, $\mathbf{%
x}:=x\boldsymbol{1}=(x,\dots ,x)$ is the $N$-dimensional vector with all
components equals to $x$. When comparing multivariate positions, all the
inequalities are to be intended componentwise, in particular, $\mathcal{B}%
^{+\!}=\mathcal{B}^{+}(\mathbb{R}^{N})$ is the set of functions in $\mathcal{%
B}=\mathcal{B}(\mathbb{R}^{N})$ with values in $[0,+\infty )^N$. A set $%
\mathcal{A\subseteq B}(\mathbb{R})$ is called monotone if $x\geq y\in 
\mathcal{A}\Rightarrow x\in \mathcal{A}$.
\end{notation}

Unless otherwise specified, in the remainder of the paper the following
assumption holds true.

\begin{assmp}
\label{ass:standing}

\begin{enumerate}
\item $\mathcal{A\subseteq B}(\mathbb{R})$ is monotone and $0\in \mathcal{A}$%
; $\mathcal{G\subseteq B}$ with $\mathbf{0}\in \mathcal{G}$.

\item The aggregation function $\Lambda :\mathbb{R}^{N}\rightarrow \mathbb{R}
$ is increasing, with respect to the componentwise order, concave and $%
\Lambda (\mathbf{0})=0$.

\item The set $\Lambda ^{-1}(\mathcal{A})-\mathcal{G}=\left\{ X\in \mathcal{B%
}\mid \Lambda (X+g)\in \mathcal{A}\text{ for some }g\in \mathcal{G}\right\} $
is convex.
\end{enumerate}
\end{assmp}

The monotonicity of the acceptance set $\mathcal{A}$ is standard in the
context of univariate risk measures and the conditions on $\Lambda$ are also
typical in the theory of multivariate risk measures. Notice that the
aggregation function is not required to be strictly increasing nor strictly
concave. Given the first two, the third condition holds when both $\mathcal{A%
}$ and $\mathcal{G}$ are convex.

\begin{exmp}
\label{agg:fun:exmp}Consider concave increasing functions $u,u_{i}:\mathbb{R}%
\rightarrow \mathbb{R}$ satisfying $u(0)=u_{i}(0)=0 ~\forall i \in \{1,
\dots, N\}$. The aggregation functions 
\begin{eqnarray}
\Lambda (x) &:&=\alpha u\bigg(\sum_{i=1}^{N}x^{i}\bigg)+\sum_{i=1}^{N}\alpha
_{i}u_{i}(x^{i}),\text{ for }\alpha ,\alpha _{i}\geq 0 ~\forall i \in \{1,
\dots, N\},  \label{utility} \\
\Lambda (x) &:&=-\sum_{i=1}^{N}\alpha _{i}(x^{i})^{-},\text{ for }\alpha
_{i}\geq 0 ~\forall i \in \{1, \dots, N\} ,  \label{negative utility}\\
\Lambda (x)&:&=\max_{y\in \mathbb{R}_{-}^{N},b\in \mathbb{R}_{-}^{N}\ :\
		x_{i}\geq b_{i}+y_{i}-\sum_{j=1}^{N}\Pi _{ji}y_{j}}\left\{
	\sum_{i=1}^{N}y_{i}+\gamma \sum_{i=1}^{N}b_{i}\right\},\text{ for }\gamma>1 \label{network}
\end{eqnarray}%
satisfy Assumption \ref{ass:standing}. The function in (\ref{utility}) has
been frequently used in the literature on systemic risk measures with either 
$\alpha=0$ or $\alpha_i=0 ~\forall i \in \{1, \dots, N\}$, see e.g.\ \cite%
{ACDP,systemic1}. The function in (\ref{negative utility}) corresponds to
considering the aggregate position as the sum of the debts of the single
units, if $\alpha _{i}=1 ~\forall i \in \{1, \dots, N\}$. Finally, the function in (\ref{network}) is derived from a network model where $\Pi _{ji}$ is the fraction of the total debt of firm $j$ owed to firm 
$i$ and $\gamma >1$ is a parameter balancing the trade off between capital
injection and reduction of mutual debts, see \cite{ChenIyengarMoallemi} for more details and examples.
\end{exmp}

We now introduce the functional analytic setting that allows us to prove a
robust dual representation for $\rho $. Let $Z:\Omega \rightarrow \lbrack
1,+\infty )$ be a continuous function with compact sublevel sets $\left\{
\omega \in \Omega :Z(\omega )\leq z\right\} $ for all $z\in \mathbb{R}$. Let 
$\mathcal{B}_{Z}$ be the set of functions $X=(X^{1},\dots ,X^{N})\in 
\mathcal{B}$ such that $\frac{X^{i}}{Z}$ is bounded for all $i=1,\dots ,N$.
The set of continuous functions in $\mathcal{B}_{Z}$ is called $C_{Z}$,
while $U_{Z}$ is the set of upper semicontinuous functions in $\mathcal{B}%
_{Z}$. Their univariate counterparts are $\mathcal{B}_{Z}(\mathbb{R}),C_{Z}(%
\mathbb{R}),U_{Z}(\mathbb{R})$. We let $ca_{Z}$ be the space of $N$%
-dimensional vectors of Borel measures $\mu =(\mu ^{1},\ldots ,\mu ^{N})$
such that $\int_{\Omega }{Zd\mu ^{i}}<+\infty $, for every $i=1,\ldots ,N$.
We finally form a dual pair $\left( \mathcal{B}_{Z},ca_{Z},\left\langle
~,~\right\rangle \right) $ with 
\begin{equation*}
\left\langle X,\mu \right\rangle :=\sum_{i=1}^{N}\int_{\Omega }{X^{i}d\mu
^{i}},\qquad X\in \mathcal{B}_{Z},\ \mu \in ca_{Z}.
\end{equation*}%
We let $ca_{Z}^{+}$ be the positive cone in $ca_{Z}$ and observe that $%
ca_{Z}^{+}$ contains the subset $\mathcal{P}_{Z}$ of $N$-dimensional vectors
of probability measures $\mathbb{P}=(\mathbb{P}^{1},\dots ,\mathbb{P}^{N})$
such that $\mathbb{E}^{\mathbb{P}^{i}}Z<+\infty $ for $i=1\dots ,N$. For a
functional $f$ on $C_{Z}$ we define 
\begin{equation}
f^{\ast }(\mu ):=\sup_{X\in C_{Z}}\left\{ \left\langle X,\mu \right\rangle
-f(X)\right\} ,\qquad \mu \in ca_{Z},  \label{convex conjugate}
\end{equation}%
which is the convex conjugate of $f$ with respect to the dual system $%
(C_{Z},ca_{Z}).$ 

In Theorem \ref{main:thm:2} below, we prove that $\rho $ admits a dual
representation if and only if a certain no arbitrage condition holds. The
theorem holds under the following assumption, which is essentially requiring
that the set of achievable market payoffs $\mathcal{G}$ is rich enough.

\begin{assmp}
\label{ass:Grich} There exists $\gamma \leq 0$ such that 
\begin{equation}
\forall n\in \mathbb{N},\ \exists z\in \lbrack 0,+\infty ),\text{ }\exists
g\in \mathcal{G}\text{ such that }\Lambda \left( \left[ \frac{\gamma }{N}+%
\frac{1}{n}-n(Z-z)^{+}\right] \mathbf{1}+g\right) \in \mathcal{A}. 
\tag*{(A)}  \label{ip:1:main:thm:2}
\end{equation}
\end{assmp}

This condition is a multivariate version of condition (2.1) in \cite%
{Che:Kup:Tan}; it is satisfied for example when $\Omega$ is compact or when,
in the market described by $\mathcal{G}$, options which are sufficiently
out-of-the-money are available at a sufficiently small price. We give some
other sufficient conditions in Proposition~\ref{suff:cond:ip:1:main:thm:2}
below. By definition of $\rho $, condition \ref{ip:1:main:thm:2} implies
that for all $n\in \mathbb{N}$ there exists $z_{n}\in \mathbb{R_{+}}$ such
that $\rho (-n(Z-z_{n})^{+}\mathbf{1})\leq \frac{N}{n}+\gamma $.

\begin{thm}
\label{main:thm:2} Under Assumption \ref{ass:Grich} the following are
equivalent:

\begin{enumerate}
\item \label{N:no:arb:cond} $m\in \mathbb{R}^{N},$ $\sum m^{i}<\gamma
\Rightarrow \nexists g\in \mathcal{G}:\Lambda (m+g)\in \mathcal{A}.$

\item \label{N:gen:mart:cond} There exists $\mathbb{Q}=(\mathbb{Q}%
_{1},\ldots,\mathbb{Q}_{N})\in \mathcal{P}_{Z}$ such that $\sum_{i=1}^{N}%
\mathbb{E}^{\mathbb{Q}_{i}}[X^{i}]-\gamma \geq 0$ for all $X\in C_{Z}$
satisfying $\Lambda (X+g)\in \mathcal{A}$ for some $g\in \mathcal{G}$.

\item \label{N:dual:repr:phi} $\rho $ is real valued on $\mathcal{B}_{Z},$ $%
\rho (0)=\gamma $ and 
\begin{equation*}
\rho (X)=\max_{\mathbb{Q}=(\mathbb{Q}_{1},\ldots ,\mathbb{Q}_{N})\in 
\mathcal{P}_{Z}}\left\{ \sum_{i=1}^{N}\mathbb{E}^{\mathbb{Q}%
_{i}}[-X^{i}]-\rho ^{\ast }(-\mathbb{Q})\right\} \text{,\quad }X\in C_{Z}.
\end{equation*}%
%
%
%
%
%
\end{enumerate}

If in addition one has 
\begin{equation}
\rho (X)=\inf_{Y\in C_{Z},Y\leq X}{\rho (Y)}\text{ for all }X\in U_{Z},
\label{ip:2:main:thm:2}
\end{equation}%
then \ref{N:no:arb:cond}--\ref{N:dual:repr:phi} are also equivalent to each
one of the following two conditions:

\begin{enumerate}
\setcounter{enumi}{3}
\item \label{N:gen:mart:cond:2}There exists $\mathbb{Q}=(\mathbb{Q}%
_{1},\ldots,\mathbb{Q}_{N})\in \mathcal{P}_{Z}$ such that $\sum_{i=1}^{N}%
\mathbb{E}^{\mathbb{Q}_{i}}[X^{i}]-\gamma \geq 0$ for all $X\in U_{Z}$
satisfying $\Lambda (X+g)\in \mathcal{A}$ for some $g\in \mathcal{G}$.

\item \label{N:dual:repr:phi:2} $\rho $ is real valued on $\mathcal{B}_{Z},$ 
$\rho (0)=\gamma $ and 
\begin{equation*}
\rho (X)=\max_{\mathbb{Q}=(\mathbb{Q}_{1},\ldots ,\mathbb{Q}_{N})\in 
\mathcal{P}_{Z}}\left\{ \sum_{i=1}^{N}\mathbb{E}^{\mathbb{Q}%
_{i}}[-X^{i}]-\rho ^{\ast }(-\mathbb{Q})\right\} \text{,\quad }X\in U_{Z}.
\end{equation*}
\end{enumerate}
\end{thm}

Before giving the proof of Theorem \ref{main:thm:2}, whose technical parts
are postponed to Section~\ref{App:analytical}, we comment on its statement.
Recall that $\gamma \leq 0$ is given and fixed. 

Condition \ref{N:no:arb:cond} excludes a situation that we call a \emph{%
regulatory arbitrage opportunity}, namely a situation where it is possible
to make an initial position $m\in \mathbb{R}^{N}$ such that $%
\sum_{i=1}^{N}m^{i}<\gamma $ acceptable by simply adding an achievable payoff $g\in 
\mathcal{G}$ that is obtained by trading at zero cost in the financial
market. In particular, absence of regulatory arbitrage opportunities implies that
none of the $N$ agents can achieve a \emph{(model independent) market
arbitrage opportunity}, namely, that the set of achievable payoffs $\mathcal{%
G}$ cannot contain elements of the form $(0,\cdots ,g^{i},\cdots ,0)$ with $%
g^{i}(\omega )\geq \varepsilon >0$ for every $\omega \in \Omega $.

Condition \ref{N:gen:mart:cond} provides information regarding the existence
of an evaluation measure $\mathbb{Q}$. If it is possible to make a position $%
X\in C_{Z}$ acceptable by adding an achievable payoff $g\in \mathcal{G}$,
then the evaluation of $X$ given by $\mathbb{Q}$ (and adjusted by $\gamma $)
must be non negative. In particular, for the case $\gamma =0$ such an
evaluation of $X$ must be non-negative without adjustments. When $\mathcal{G}
$ is a linear space of functions in $C_{Z}$, we can write $\Lambda
(kg+(-kg))=\Lambda (0)=0\in \mathcal{A}$ for any $k\in \mathbb{R}$ and $g\in 
\mathcal{G}\cap C_{Z}$. Then, condition \ref{N:gen:mart:cond} implies that 
\begin{equation*}
\sum \mathbb{E}^{\mathbb{Q}^{i}}[g^{i}]=0\text{ for all }g\in \mathcal{G}%
\cap C_{Z}.
\end{equation*}%
Using the terminology of \cite{systemic2}, the probability vector $\mathbb{Q}
$ is called \emph{fair}. If, in addition, $\mathcal{G}$ contains vectors
with only one non-zero components then $\mathbb{Q}$ is a vector of
martingale measures, i.e., $\mathbb{E}_{\mathbb{Q}^{i}}[g^{i}]=0$ for every $%
(0,\cdots ,g^{i},\cdots ,0)\in \mathcal{G}$.

Finally, condition \ref{N:dual:repr:phi} is the usual dual representation of
the Fenchel-Moreau type. Following the original interpretation of \cite{Artzner:Coherent} each $\mathbb{Q}\in \mathcal{P}_{Z}$ is a plausible model
for the risk $X$ and it is called a \emph{generalized scenario}. The term $%
\rho ^{\ast }(-\mathbb{Q})$ has the role to penalize scenarios which are
less plausible, possibly by an infinite amount. The value $\rho (X)$ is then
the worst-case expectation across all plausible scenarios, suitably
penalized. We refer to \cite[Chapter 4]{FS2016} and \cite[Chapter 8]{McNeil}, for more details on the relevance of the dual representation for the
univariate case and to \cite{ararat:birgit} and \cite{systemic:arduca:2021} for a comprehensive
study of dual representations for systemic risk measures in the
non-robust setting. 
\begin{rem}\label{rem:FriSca}
	Using an argument similar to \cite[Proposition 3.9 and 3.11]{FriSca06}, the dual representation in item \ref{N:dual:repr:phi} could be reformulated as%
	\begin{equation*}
		\rho (X)=\max_{\mathbb{Q}=(\mathbb{Q}_{1},\ldots ,\mathbb{Q}_{N})\in 
			\mathcal{P}_{Z}\cap Bar_{\Lambda ^{-1}(\mathcal{A)}}\cap (-Bar_{\mathcal{G}%
			})}\left\{ \sum_{i=1}^{N}\mathbb{E}^{\mathbb{Q}_{i}}[-X^{i}]-\sigma
		_{\Lambda ^{-1}(\mathcal{A})}(\mathbb{Q})-\sigma _{\mathcal{G}}(-\mathbb{Q}%
		)\right\} \text{,\quad }X\in C_{Z},
	\end{equation*}%
where $Bar_{\mathbb{A}}:=\left\{ \mathbb{Q\in }ca_{Z}^{+}\mid \sigma _{%
		\mathbb{A}}(\mathbb{Q})<+\infty \right\} $ is the domain of finiteness of
	the support function $\sigma _{\mathbb{A}}(\mathbb{Q}):=\sup_{W\in \mathbb{A}%
	}\sum_{i=1}^{N}\mathbb{E}^{\mathbb{Q}_{i}}[-W^{i}]$ of a set $\mathbb{A}%
	\mathcal{\subseteq B}(\mathbb{R}^{N}).$ Such a formula emphasizes the role of the defining ingredients of 
	$\rho $ in the penalty function $\rho^*$. A thorough and extensive analysis - in the non robust setting - of such decomposition of $\rho^*$ can be found in \cite[Section 3.3]{systemic:arduca:2021}
\end{rem}

We next prove Theorem \ref{main:thm:2}. On the technical side we need to
adapt some results of \cite{Che:Kup:Tan} to the multivariate case and we
provide them in Section~\ref{App:analytical}. Note that, differently from 
\cite{Che:Kup:Tan} we allow $\mathcal{A}$ to contain purely negative
acceptable position. This is important in the context of systemic risk in
order to include examples of aggregation functions which always yield
non-positive random variables, as in Example \ref{agg:fun:exmp} equation %
\eqref{negative utility}. We start with an easy observation.

\begin{lem}
\label{lem:cashadd} The map $\rho :\mathcal{B}\rightarrow \lbrack -\infty
,+\infty ]$ defined in \eqref{def:sysrho} is monotone decreasing, convex and
(systemically) cash additive, namely 
\begin{equation*}
\rho (X+c)=\rho (X)-\sum_{i=1}^{N}c^{i}\text{, for all }X\in \mathcal{B} 
\text{ and }c\in \mathbb{R}^{N}.
\end{equation*}
\end{lem}

\begin{proof}
Since $\Lambda $ is increasing and $\mathcal{A}$ is monotone, the set $%
\Lambda ^{-1}(\mathcal{A})-\mathcal{G}$ is monotone and then monotonicity of $\rho$ is easily checked.
Regarding convexity, for $m,n\in \mathbb{R}^{N}$
such that $m+X,\ n+Y \in \Lambda ^{-1}(\mathcal{A})-\mathcal{G}$, one gets 
$
\lambda m+(1-\lambda )n+\left( \lambda X+(1-\lambda )Y\right) \in \Lambda
^{-1}(\mathcal{A})-\mathcal{G},
$
for all $\lambda\in[0,1]$, by the convexity of $\Lambda ^{-1}(\mathcal{A})- 
\mathcal{G}$ . Hence, 
\begin{equation*}
\rho \left( \lambda X+(1-\lambda )Y\right) \leq \lambda \sum_{i=1}^{N}{m^{i}}
+(1-\lambda )\sum_{i=1}^{N}{n^{i}}
\end{equation*}
and convexity now follows by taking the infimum over $m$ and $n$ satisfying $%
m+X,$ $n+Y\in \Lambda ^{-1}(\mathcal{A})-\mathcal{G}$ on the right-hand
side. The cash additivity property is trivial.
\end{proof}

\begin{proof}[Proof of Theorem~\protect\ref{main:thm:2}]
\ref{N:dual:repr:phi}$\implies $\ref{N:gen:mart:cond}. Let $X\in C_{Z}\cap
(\Lambda ^{-1}(\mathcal{A})-\mathcal{G})$. As $\Lambda (X+g)\in \mathcal{A}$
for some $g\in \mathcal{G}$, the definition of $\rho $ implies that $\rho
(X)\leq 0$. By the dual formula for $\rho $ in item \ref{N:dual:repr:phi},
one obtains 
\begin{equation*}
\gamma =\rho (0)=\max_{\mathbb{Q}\in \mathcal{P}_{Z}}{-\rho ^{\ast }(-%
\mathbb{Q})}=-\min_{\mathbb{Q}\in \mathcal{P}_{Z}}{\rho ^{\ast }(-\mathbb{Q})%
},
\end{equation*}%
implying the existence of $\hat{\mathbb{Q}}\in \mathcal{P}_{Z}$ such that $%
\rho ^{\ast }(-\hat{\mathbb{Q}})=-\gamma $. Using the definition of $\rho
^{\ast }$ and $\rho (X)\leq 0$, we have 
\begin{align*}
-\gamma =\rho ^{\ast }(-\hat{\mathbb{Q}})& \geq \sup_{X\in C_{Z}\cap
(\Lambda ^{^{-1}}\!(\mathcal{A})-\mathcal{G})}\left\{ \sum_{i=1}^{N}{\mathbb{%
E}^{\hat{\mathbb{Q}}^{i}}[-X^{i}]}-\rho (X)\right\} \\
& \geq \sup_{X\in C_{Z}\cap (\Lambda ^{^{-1}}\!(\mathcal{A})-\mathcal{G}%
)}\left\{ \sum_{i=1}^{N}{\mathbb{E}^{\hat{\mathbb{Q}}^{i}}[-X^{i}]}\right\} ,
\end{align*}%
from which item \ref{N:gen:mart:cond} follows readily.

\ref{N:gen:mart:cond}$\implies $\ref{N:no:arb:cond}. Let $m\in\mathbb{R}^N$, 
$g\in\mathcal{G}$ such that $\Lambda (m+g)\in \mathcal{A}$. By item \ref%
{N:gen:mart:cond}, there exists $\mathbb{Q}\in \mathcal{P}_{Z}$ such that 
$
\sum_{i=1}^Nm^i-\gamma=\sum_{i=1}^N E^{\mathbb{Q}_{i}}[m^{i}]-\gamma \geq 0,
$
from which item \ref{N:no:arb:cond} follows.

\ref{N:no:arb:cond}$\implies $\ref{N:dual:repr:phi}. Set $\Phi (X):=\rho
(-X) $ and notice that $\Phi $ is monotone increasing and convex. By the
cash additivity property of Lemma \ref{lem:cashadd}, $\Phi (m)=\Phi
(0)+\sum_{i=1}^{N}{m^{i}}$, for all $m\in \mathbb{R}^{N}$. We now show that $%
\Phi (0)=\gamma$.
By item \ref{N:no:arb:cond}, if $m\in {\mathbb{R}}^{N}$ and $\Lambda
(m+g)\in \mathcal{A}$ then $\sum m^{i}\geq \gamma $, which implies $\Phi
(0)=\rho (0)\geq \gamma $. Moreover, by Condition~\ref{ip:1:main:thm:2}, for
all $n\in \mathbb{N}$ and $z=z(n)$ large enough the following holds 
\begin{equation*}
\gamma \leq \Phi (0)\leq \Phi \left( n(Z-z)^{+}\mathbf{1}\right) \leq \frac{N%
}{n}+\gamma .
\end{equation*}%
This implies $\rho (0)=\Phi (0)=\gamma $. We now show that $\Phi $ is real
valued on $\mathcal{B}_{Z}$. Let $X\in \mathcal{B}_{Z}$ and $k\in \mathbb{N}$
such that $-\frac{1}{2}kZ\mathbf{1}\leq X\leq \frac{1}{2}kZ\mathbf{1}$.
Using \ref{ip:1:main:thm:2}, there exists $z=z(k)\in \mathbb{R}_{+}$ large
enough such that $\Phi (k(Z-z)^{+}\mathbf{1})\leq \frac{N}{k}+\gamma <\infty 
$. By cash additivity, $\Phi (kz\mathbf{1})=Nkz<\infty $. By monotonicity
and convexity of $\Phi $ we then deduce 
\begin{eqnarray*}
\Phi (X) &\leq &\Phi \left( \frac{k}{2}Z\mathbf{1}\right) =\Phi \left( \frac{%
1}{2}k(Z-z)\mathbf{1}+\frac{1}{2}kz\mathbf{1}\right) \leq \frac{1}{2}\Phi
\left( k(Z-z)\mathbf{1}\right) +\frac{1}{2}\Phi (kz\mathbf{1}) \\
&\leq &\frac{1}{2}\Phi \left( k(Z-z)^{+}\mathbf{1}\right) +\frac{1}{2}\Phi
(kz\mathbf{1})<\infty .
\end{eqnarray*}%
By convexity, for any $Y$ we have $\gamma =\Phi (0)\leq \frac{1}{2}\left(
\Phi (Y)+\Phi (-Y)\right) $, hence $\Phi (X)\geq \Phi (-\frac{k}{2}Z\mathbf{1%
})\geq 2\gamma -\Phi (\frac{k}{2}Z\mathbf{1})>-\infty $. The conclusion now
follows as in the univariate case. Using \cite[Lemma~$\mathbf{A.2}$]%
{Che:Kup:Tan}, Condition~\ref{ip:1:main:thm:2} implies that also Condition~%
\eqref{dual:repr:thm:ip} in Theorem~\ref{dual:repr:thm} below holds. Using
Theorem~\ref{dual:repr:thm} we obtain the representation $\rho (X)=\Phi
(-X)=\max_{\mu \in ca_{Z}^{+}}\left\{ \left\langle -X,\mu \right\rangle
-\Phi ^{\ast }(\mu )\right\} $ for all $X\in C_{Z}$. It is now sufficient to
note that $\Phi (m^{[i]})=m+\gamma $, where $m^{[i]}$ 
is the vector with the $i$-th coordinate equals to $m\in \mathbb{R}$ and all
the others are zero, to observe that $\Phi ^{\ast }(\mu )=+\infty $ for $\mu
\in ca_{Z}^{+}\backslash \mathcal{P}_{Z}$. The desired dual representation
in item \ref{N:dual:repr:phi} follows from ${\rho ^{\ast }(-\mathbb{\mu })=}%
\Phi ^{\ast }(\mu ),$ $\mu \in ca_{Z}^{+}.$

This proves \ref{N:no:arb:cond}$\iff$\ref{N:gen:mart:cond}$\iff$\ref%
{N:dual:repr:phi}. Suppose now Condition~\eqref{ip:2:main:thm:2} holds. The
implications \ref{N:dual:repr:phi:2}$\implies $\ref{N:gen:mart:cond:2} and %
\ref{N:gen:mart:cond:2}$\implies$\ref{N:gen:mart:cond} are easily seen and
similar to above. We conclude by proving \ref{N:dual:repr:phi}$\implies$\ref%
{N:dual:repr:phi:2}. Consider once again $\Phi(X):=\rho(-X)$. By \eqref{ip:2:main:thm:2}, $\Phi$ satisfies condition \ref{cond:3:dual:U} of Theorem~\ref{dual:repr:thm:U}
below. From item \ref{N:dual:repr:phi} and Condition~\ref{ip:1:main:thm:2}, $%
\Phi$ is a real-valued, increasing convex functional on $\mathcal{B}_{Z}$
that satisfies Condition~(\ref{zero:cont:ip}) below. The desired
representation follows from Theorem~\ref{dual:repr:thm:U} Condition \ref%
{cond:1:dual:U} and the fact that $\Phi^{*}(\mu)=+\infty$ for $\mu\in
ca_{Z}^{+}\backslash\mathcal{P}_{Z} $, as in \ref{N:no:arb:cond}$\implies$%
\ref{N:dual:repr:phi}.
\end{proof}
We end this section by discussing when the assumptions of Theorem \ref%
{main:thm:2} are satisfied. A trivial case for \ref{ip:1:main:thm:2}
is when $\Omega $ is a compact subset of $\mathbb{R}^{T(J+1)}$, as $Z$ is a
continuous function. On the other hand, Assumption~\eqref{ip:2:main:thm:2}
is verified if $\rho $ is continuous from below, since for all $X\in U_{Z}$,
there exists a sequence $\left( Y_{n}\right) _{n\in \mathbb{N}}$ in $C_{Z}$
such that $Y_{n}\uparrow X$. We next present another sufficient condition
for \ref{ip:1:main:thm:2}.

\begin{prop}
\label{suff:cond:ip:1:main:thm:2} 
Consider a continuous and strictly increasing aggregation function $\Lambda$ and $\alpha:\mathcal{P}_{Z}\to\mathbb{R}_{+}\cup\{+\infty\}$ satisfying:

\begin{enumerate}
	
	\item \label{item1A} $\inf_{\mathbb{P}\in\mathcal{P}_{Z}}\alpha(\mathbb{P}%
	)=0 $.
	
	\item \label{item2A} $\alpha(\mathbb{P})\geq\frac{1}{N}\sum_{i=1}^{N}\mathbb{%
		E}^{\mathbb{P}^{i}}\beta(Z) \text{ for all } \mathbb{P}\in\mathcal{P}_{Z}$,
	where $\beta:[1, +\infty)\to\mathbb{R}$ is an increasing function with the
	property $\lim_{x\to+\infty}\frac{\beta(x)}{x}=+\infty$.
\end{enumerate}
 
Condition~$\ref{ip:1:main:thm:2}$ holds for the following acceptance set:
\begin{equation}
	\mathcal{A}:=\left\{ X\in \mathcal{B}_{Z}(\mathbb{R}):\sum_{i=1}^{N}\mathbb{E%
	}^{\mathbb{P}^{i}}[X]+\alpha (\mathbb{P})\geq 0\text{ for all }\mathbb{P}\in 
	\mathcal{P}_{Z}\right\} +\mathcal{B}^{+}(\mathbb{R}).  \label{def:A:alpha}
\end{equation}
\end{prop}

\begin{proof}
The proof is similar to that of \cite[Proposition~$\mathbf{2.2}$]%
{Che:Kup:Tan} so we only present the main differences. By passing to the
lower convex hull, $\beta$ in \ref{item2A} can be assumed convex and
Jensen's inequality yields 
\begin{equation*}
\alpha(\mathbb{P})\geq\frac{1}{N}\sum_{i=1}^{N}\mathbb{E}^{\mathbb{P}
^{i}}\beta(Z)\geq\beta\left(\sum_{i=1}^{N}\mathbb{E}^{\mathbb{P}
^{i}}Z\right).
\end{equation*}
Consider now the sets $T_{a}:=\left\{\mathbb{P}\in\mathcal{P} _{Z} :
\sum_{i=1}^{N}\mathbb{E}^{\mathbb{P}^{i}}\beta(Z)\leq a\right\}, a\in 
\mathbb{R}$. By using that $\beta$ is increasing and $Z\ge 1$, we deduce
that, for each $\mathbb{P}^i$ with $\mathbb{P}\in T_{a}$, it holds $\mathbb{E%
}^{\mathbb{P}^{i}}\beta(Z)\leq a - (N-1)\beta(1)$ for all $i=1, \dots, N$.
We deduce that the one-dimensional projections $\pi^{i}( T_{a})$ are $\sigma(%
\mathcal{P} _{Z}, C_{Z})$-compact. Hence, by identifying $(C_{Z}(\mathbb{R}%
), ca_{Z}^{+}(\mathbb{R}))$ with $(C_{b}( \mathbb{R}), ca^{+}(\mathbb{R}))$,
one can use Prokhorov's theorem to obtain the compactness of $\pi^{i}(
T_{a}) $. It follows that the sets $\left\{\mathbb{P}\in\mathcal{P}_{Z} :
\alpha( \mathbb{P})\leq a\right\}, a\in\mathbb{R}$ are relatively $\sigma(%
\mathcal{P} _{Z}, C_{Z})$-compact. Define now $\tau:\mathcal{B}_{Z}(\mathbb{R%
})\to\mathbb{R}$ as 
\begin{equation*}
\tau(X) :=\sup_{\mathbb{P}\in\mathcal{P}_{Z}}\left(\sum_{i=1}^{N}\mathbb{E}^{%
\mathbb{P}^{i} }[-X] - \alpha(\mathbb{P})\right),
\end{equation*}
so that one can rewrite $\mathcal{A}=\left\{X\in\mathcal{B}_{Z}(\mathbb{R})
: \tau(X)\leq0\right\} + \mathcal{B} ^{+}(\mathbb{R})$. Fix $n\in\mathbb{N}$
and take $X_z:=(\frac{1}{n}-n(Z-z)^{+})\mathbf{1}\in\mathcal{B}_Z$. Using
that $\Lambda$ is continuous and non-decreasing, $-\Lambda(X_z)\downarrow-%
\Lambda(\frac{1}{n}\mathbf{1})$ as $z\to\infty$. Using the compactness of
the sets $\left\{\mathbb{P}\in\mathcal{P}_{Z} : \alpha( \mathbb{P})\leq
a\right\}$, we can now apply \cite[Lemma~$\mathbf{A.4}$]{Che:Kup:Tan}, which
readily extends to the multivariate case, to deduce that $%
\tau(\Lambda(X_z))\downarrow \tau(\Lambda(\frac{1}{n}\mathbf{1}))=-N\Lambda(%
\frac{1}{n}\mathbf{1})$. The last equality follows from \ref{item1A} and,
using that $\Lambda$ is strictly increasing in this proposition, we have $%
-N\Lambda(\frac{1}{n}\mathbf{1})<0$. We deduce that for $z$ large enough, $%
\Lambda(X_z)$ is acceptable. Since $n\in\mathbb{N}$ was arbitrary, Condition~%
$\ref{ip:1:main:thm:2}$ is satisfied with the choice of $\gamma=0$.
\end{proof}

\subsection{Different aggregation functions}

\label{sec:Gamma} As discussed in the introduction, an alternative way of
measuring market-adjusted systemic risk consists in the use of a second
aggregation function $\Gamma :\mathbb{R}^{N}\rightarrow \mathbb{R}$ for the
payoffs of the trading strategies. In this scenario, $\rho_{\Gamma}$ is
defined as in $(\ref{def:sysrho2})$ and Theorem~\ref{main:thm:2} changes
consequently. In particular, in this subsection the following Assumption replaces
Assumption \ref{ass:standing}:

\begin{assmp}
\label{ass:standing:Gamma}

\begin{enumerate}
\item $\mathcal{A\subseteq B}(\mathbb{R})$ is monotone and $0\in \mathcal{A}$
; $\mathcal{G\subseteq B}$ with $\mathbf{0}\in \mathcal{G}$.

\item The aggregation functions $\Lambda:\mathbb{R}^{N}\rightarrow \mathbb{R}
$ and $\Gamma:\mathbb{R}^{N}\rightarrow \mathbb{R} $ are increasing with
respect to the componentwise order, concave and $\Lambda (\mathbf{0})=\Gamma
(\mathbf{0})=0$.

\item The set $\Lambda^{-1}(\mathcal{A} - \Gamma(\mathcal{G})) = \left\{
X\in \mathcal{B }\mid \Lambda (X) +\Gamma(g) \in \mathcal{A}\text{ for some }%
g\in \mathcal{G}\right\}$ is convex.
\end{enumerate}
\end{assmp}

As in the previous case, $\rho_{\Gamma}$ is still a monotone decreasing,
convex and (systematically) cash additive map, and the same functional
analytic setting is considered.

\begin{assmp}
\label{ass:Grich:Gamma} There exists $\gamma \leq 0$ such that 
\begin{equation}
\forall n\in \mathbb{N},\ \exists z\in \lbrack 0,+\infty ),\text{ }\exists
g\in \mathcal{G}\text{ such that }\Lambda \left( \left[ \frac{\gamma }{N}+ 
\frac{1}{n}-n(Z-z)^{+}\right] \mathbf{1}\right)+\Gamma(g) \in \mathcal{A}. 
\tag*{(A)}  \label{ip:1:main:thm}
\end{equation}
\end{assmp}

\begin{thm}
\label{main:thm} Under Assumption~\ref{ass:Grich:Gamma} the following are
equivalent:

\begin{enumerate}
\item \label{N:no:arb:cond:Gamma} $m\in \mathbb{R}^{N},$ $\sum m^{i}<\gamma
\Rightarrow \nexists g\in \mathcal{G}:\Lambda(m)+\Gamma(g)\in \mathcal{A}.$

\item \label{N:gen:mart:cond:Gamma} There exists $\mathbb{Q}=(\mathbb{Q}%
_{1},\ldots, \mathbb{Q}_{N})\in \mathcal{P}_{Z}$ such that $\sum_{i=1}^{N}%
\mathbb{E}^{ \mathbb{Q}_{i}}[X^{i}]-\gamma \geq 0$ for all $X\in C_{Z}$
satisfying $\Lambda (X)+\Gamma(g)\in \mathcal{A}$ for some $g\in \mathcal{G} 
$.

\item \label{N:dual:repr:phi:Gamma} $\rho _{\Gamma }$ is real valued on $%
\mathcal{B}_{Z},$ $\rho _{\Gamma }(0)=\gamma $ and 
\begin{equation*}
\rho _{\Gamma }(X)=\max_{\mathbb{Q}=(\mathbb{Q}_{1},\ldots ,\mathbb{Q}%
_{N})\in \mathcal{P}_{Z}}\left\{ \sum_{i=1}^{N}\mathbb{E}^{\mathbb{Q}%
_{i}}[-X^{i}]-\rho _{\Gamma }^{\ast }(-\mathbb{Q})\right\} \text{,\quad }%
X\in C_{Z}.
\end{equation*}%
\end{enumerate}
If, in addition, one has 
\begin{equation}
\rho_{\Gamma} (X)=\inf_{Y\in C_{Z},Y\leq X}{\rho_{\Gamma} (Y)}\text{ for all 
}X\in U_{Z},  \label{ip:2:main:thm}
\end{equation}
then \ref{N:no:arb:cond:Gamma}--\ref{N:dual:repr:phi:Gamma} are also
equivalent to each one of the following two conditions:

\begin{enumerate}
\setcounter{enumi}{3}

\item \label{N:gen:mart:cond:2:Gamma}There exists $\mathbb{Q}=(\mathbb{Q}%
_{1},\ldots, \mathbb{Q}_{N})\in \mathcal{P}_{Z}$ such that $\sum_{i=1}^{N}%
\mathbb{E}^{ \mathbb{Q}_{i}}[X^{i}]-\gamma \geq 0$ for all $X\in U_{Z}$
satisfying $\Lambda (X)+\Gamma(g)\in \mathcal{A}$ for some $g\in \mathcal{G} 
$.

\item \label{N:dual:repr:phi:2:Gamma} $\rho _{\Gamma }$ is real valued on $%
\mathcal{B}_{Z},$ $\rho _{\Gamma }(0)=\gamma $ and 
\begin{equation*}
\rho _{\Gamma }(X)=\max_{\mathbb{Q}=(\mathbb{Q}_{1},\ldots ,\mathbb{Q}%
_{N})\in \mathcal{P}_{Z}}\left\{ \sum_{i=1}^{N}\mathbb{E}^{\mathbb{Q}%
_{i}}[-X^{i}]-\rho _{\Gamma }^{\ast }(-\mathbb{Q})\right\} \text{,\quad }%
X\in U_{Z}.
\end{equation*}
\end{enumerate}
\end{thm}

The proof of this theorem is analogous to the one of Theorem~\ref{main:thm:2}
and is omitted. 

\section{Multivariate analytical results and proofs}

\label{App:analytical} 
This section is dedicated to the extension of the analytical results in
Section~\textbf{A.1} of \cite{Che:Kup:Tan}. The space $C_{Z}$ of continuous
functions $X:\Omega \rightarrow \mathbb{R}^{N}$ such that each component of $%
X/Z$ is bounded forms a Stone vector lattice, with the partial order given
by $X\leq Y\iff X^{i}(\omega )\leq Y^{i}(\omega )~\forall \omega \in \Omega
,~\forall i=1,\dots ,N$. In \cite{Che:Kup:Tan} the cone $ca_{Z}^{+}(\mathbb{R%
})$ is endowed with the weak convergences topology $\sigma \left( ca_{Z}^{+}(%
\mathbb{R}),C_{Z}(\mathbb{R})\right) $, derived by the one described in
Chapter~$\mathbf{8}$ of \cite{bogachev2007measure}. Since a sequence of
multivariate measures $\left( \mu _{n}=(\mu _{n}^{1},\dots ,\mu
_{n}^{N})\right) _{n}$ weakly converges ($\rightharpoonup $) to a
multivariate measure $\mu =(\mu ^{1},\dots ,\mu ^{N})$ if and only if $\mu
_{n}^{i}\rightharpoonup \mu ^{i}$ for all $i=1,\dots ,N$, it is natural to
endow $ca_{Z}^{+}=\big(ca_{Z}^{+}(\mathbb{R})\big)^{N}$ with the product
topology. This choice allows for easy extensions of the analytical results
to the multivariate case, since compactness and metrizability are preserved
for a finite product of topological spaces. In fact $ca_{Z}^{+}(\mathbb{R})$
endowed with $\sigma \left( ca_{Z}^{+}(\mathbb{R}),C_{Z}(\mathbb{R})\right) $
is metrizable, as stated in the proof of \cite[Lemma~$\mathbf{A.4}$]%
{Che:Kup:Tan}: since $\Omega $ is a separable metric space, this follows by
results in \cite[Section~$\mathbf{8.3}$]{bogachev2007measure}. Given a
convex increasing functional $\Psi :C_{Z}\rightarrow \mathbb{R}\cup
\{+\infty \}$, its convex conjugate function $\Psi ^{\ast
}:ca_{Z}\rightarrow \mathbb{R}\cup \{+\infty \}$ is defined in 
\eqref{convex
conjugate}. 

We first seek for a dual representation formula for general increasing
convex functionals on $C_{Z}$.

\begin{thm}[{Multivariate version of {\protect\cite[Theorem~$\mathbf{A.1}$]%
{Che:Kup:Tan}}}]
\label{dual:repr:thm} Let $\Psi :C_{Z}\rightarrow \mathbb{R}^{N}$ be an
increasing convex functional with the property that for every $X\in C_{Z}$
there exists a constant $\varepsilon >0$ such that 
\begin{equation}
\lim_{z\rightarrow +\infty }{\Psi \big(X+\varepsilon (Z-z)^{+}\boldsymbol{1}%
\big)=\Psi (X)}.  \label{dual:repr:thm:ip}
\end{equation}%
Then 
\begin{equation}
\Psi (X)=\max_{\mu \in ca_{Z}^{+}}\left\{ \left\langle X,\mu \right\rangle
-\Psi^{\ast }(\mu )\right\} .  \label{dual:repr}
\end{equation}
\end{thm}

\begin{proof}
Fix $X\in C_{Z}$. By the definition of $\Psi^{\ast }$ it is obvious that 
\begin{equation}
\Psi (X)\geq \sup_{\mu \in ca_{Z}^{+}}\left\{ \left\langle X,\mu
\right\rangle -\Psi^{\ast }(\mu )\right\} .  \label{dual:ineq}
\end{equation}

We extend the approach of \cite[Theorem~$\mathbf{A.1}$]{Che:Kup:Tan} to this
setting. The Hahn-Banach extension theorem applied to the null functional on
the trivial subspace $\{0\}\subseteq C_{Z}$ implies the existence of an
increasing linear functional $\zeta _{X}:C_{Z}\rightarrow \mathbb{R}$
dominated by $\Psi _{X}(Y):=\Psi (X+Y)-\Psi (X).$ Moreover, one can define $%
N $ increasing linear functionals $\zeta _{X}^{i}:C_{Z}(\mathbb{R}%
)\rightarrow \mathbb{R}$ as follows: For $\varphi \in C_{Z}(\mathbb{R})$ and 
$\varphi ^{\lbrack i]}:=(0,\dots ,0,\varphi ,0,\dots ,0)\in C_{Z}$,
i.\thinspace e.\ $\varphi ^{\lbrack i]}$ is the $N$-dimensional functional
with the $i$-th component equals to $\varphi $ and all the others are zero,
we set%
\begin{equation}
\zeta _{X}^{i}(\varphi ):=\zeta _{X}(\varphi ^{\lbrack i]}).
\label{def:zeta:chi:i}
\end{equation}%
Hence, since $\zeta _{X}$ is linear, for all $Y=(Y^{1},\dots ,Y^{N})\in
C_{Z} $ we have $\zeta _{X}(Y)=\sum_{i=1}^{N}{\zeta _{X}^{i}(Y^{i})}$.
It will now be sufficient to prove that, for $\left( X_{n}\right) _{n}$ in $%
C_{Z}$ satisfying $X_{n}\downarrow 0$, there exists a constant $\eta >0$
such that $\Psi _{X}(\eta X_{n})\downarrow 0$. This would imply that $\zeta
_{X}(X_{n})\downarrow 0$, as $\zeta _{X}$ is linear and dominated by $\Psi
_{X}$. So, given $(\varphi _{n})_{n}$ in $C_{Z}(\mathbb{R})$ satisfying $%
\varphi _{n}\downarrow 0$, by (\ref{def:zeta:chi:i}) it follows that $\zeta
_{X}^{i}(\varphi _{n})=\zeta _{X}(\varphi _{n}^{[i]})\downarrow 0$ for all $%
i=1,\dots ,N$, since $(\varphi _{n}^{[i]})_{n}$ is a sequence going to $0$
in $C_{Z}$. Hence, as in the univariate case, the Daniell-Stone theorem
implies that there exists $\mu _{X}^{i}\in ca_{Z}^{+}(\mathbb{R})$ such that 
$\zeta _{X}^{i}(\varphi )=\left\langle \varphi ,\mu _{X}^{i}\right\rangle $
for all $\varphi \in C_{Z}(\mathbb{R})$. Defining $\mu _{X}:=(\mu
_{X}^{1},\dots ,\mu _{X}^{N})\in ca_{Z}^{+}$, one has, for all $Y\in C_{Z}$, 
\begin{equation*}
\zeta _{X}(Y)=\sum_{i=1}^{N}{\zeta _{X}^{i}(Y^{i})}=\sum_{i=1}^{N}{%
\left\langle Y^{i},\mu _{X}^{i}\right\rangle }=\left\langle Y,\mu
_{X}\right\rangle .
\end{equation*}

Since $\Psi_{X}(Y) \geq \zeta_{X}(Y)$ for all $Y\in C_{Z}$, by the identity $%
Y=X+Y-X$ it follows that $\left\langle X, \mu_{X}\right\rangle - \Psi(X)
\geq \left\langle X+Y, \mu_{X}\right\rangle - \Psi(X+Y)$. Moreover, as for
any $W\in C_{Z}, Y:= W-X$ is still in $C_{Z}$, it follows that 
\begin{equation*}
\left\langle X, \mu_{X}\right\rangle - \Psi(X) \geq \left\langle W,
\mu_{X}\right\rangle - \Psi(W)
\end{equation*}
for all $W\in C_{Z}$. Hence, $\Psi^{*}(\mu_{X}) = \left\langle X,
\mu_{X}\right\rangle - \Psi(X)$. This implies, along with (\ref{dual:ineq}),
the dual representation in (\ref{dual:repr}) and the maximum is attained by $%
\mu_{X}$. From now on the proof is identical to the one of \cite[Theorem~$\mathbf{A.1}$%
]{Che:Kup:Tan}, we provide it for completeness. Fix $\left(X_{n}\right)_{n}$
in $C_{Z}$ satisfying $X_{n}\downarrow0, \varepsilon>0$ such that (\ref%
{dual:repr:thm:ip}) holds and $m>0$ so that $X_{1}\leq mZ\boldsymbol{1}$.
Set $\eta=\frac{\varepsilon}{4m}$ and $\delta>0$. Assumption \ref%
{dual:repr:thm:ip} implies the existence of a $z>0$ such that $%
\Psi_{X}(\varepsilon(Z-z)^{+}\boldsymbol{1})<\delta$, and the set $\{Z\leq
2z\}$ is compact. It is important to note that each component of $X_{n}$ is
continuous for all $n$. So, there are no difficulties in the application of
Dini's lemma, which implies that 
\begin{equation*}
x_{n}:=(x_{n}^{1}, \dots, x_{n}^{N}), \text{ for } x_{n}^{i}:=\max_{\omega%
\in\left\{Z\leq2z\right\}}{X_{n}^{i}(\omega)},
\end{equation*}
are the elements of a sequence in $\mathbb{R}^{N}$ decreasing to $0$. Also,
for $x\in\mathbb{R}^{N}$, $x\mapsto\Psi_{X}(x)$ is a continuous function
since it is convex from $\mathbb{R}^{N}$ to $\mathbb{R}$, and so there
exists $n_{0}$ such that $\Psi_{X}(2\eta x_{n})\leq\delta$ for all $n\geq
n_{0}$. Now, 
\begin{equation*}
X_{n} \leq X_{n}\mathbbm{1}_{\{Z\leq2z\}} + X_{1}\mathbbm{1}_{\{Z>2z\}} \leq
x_{n}\mathbbm{1}_{\{Z\leq2z\}} + mZ\boldsymbol{1}\mathbbm{1}_{\{Z>2z\}} \leq
x_{n} + 2m(Z-z)^{+}\boldsymbol{1} \footnote{%
It is trivial to verify $mZ\mathbbm{1}_{\{Z>2z\}} \leq 2m(Z-z)^{+}$, since
if $\omega\notin\{Z>2z\}$, then $0\leq2m(Z-z)^{+}$ and, if $%
\omega\in\{Z>2z\} $, then $mZ\leq2m(Z-z) \iff Z>2z$.}
\end{equation*}
implies that $\frac{X_{n}-x_{n}}{2m} \leq (Z-z)^{+}\boldsymbol{1},$
and therefore, since $\Psi_{X}$ is increasing, 
\begin{equation*}
\Psi_{X}(2\eta(X_{n}-x_{n})) = \Psi_{X}\left(\varepsilon\frac{X_{n}-x_{n}}{2m%
}\right)\leq\delta \text{ for all } n.
\end{equation*}
This gives, by the convexity of $\Psi_{X}$ and the fact that $\Psi_{X}(0)=0$%
, that 
\begin{equation*}
\Psi_{X}(\eta X_{n}) \leq \frac{\Psi_{X}(2\eta x_{n}) + \Psi_{X}(2\eta
(X_{n} - x_{n}))}{2} \leq \delta \text{ for all } n\geq n_{0},
\end{equation*}
implying that $\Psi_{X}(\eta X_{n})\downarrow 0$. This completes the proof.
\end{proof}

A sufficient condition for Condition~(\ref{dual:repr:thm:ip}) is the
following: 
\begin{equation}  \label{zero:cont:ip}
\lim_{z\to+\infty}{\Psi\left(n(Z - z)^{+}\boldsymbol{1}\right)} = \Psi(0) 
\text{ for every } n\in\mathbb{N}.
\end{equation}

As stated in \cite[Lemma~$\mathbf{A.2}$]{Che:Kup:Tan}, that is trivially
valid in a multivariate case, Condition~(\ref{zero:cont:ip}) implies
Condition~(\ref{dual:repr:thm:ip}) for an increasing convex functional $%
\Psi:C_{Z}\to\mathbb{R}$. We next extend the dual representation to $U_{Z}$,

\begin{lem}[{Multivariate version of {\protect\cite[Lemma~$\mathbf{A.3}$]%
{Che:Kup:Tan}}}]
\label{phi*:sublev:comp} Let $\Psi:C_{Z}\to\mathbb{R}$ be an increasing
convex functional. The following hold.

\begin{enumerate}

\item\label{item1:LemmaA3} There exists an increasing convex function $\varphi:\mathbb{R}_{+}\to%
\mathbb{R}\cup\{+\infty\}$ satisfying $\lim_{x\to+\infty}\frac{\varphi(x)}{x}
= +\infty$ such that $\Psi^{*}(\mu) \geq \varphi\left( \left\langle Z%
\boldsymbol{1}, \mu\right\rangle\right)$ for all $\mu\in ca_{Z}^{+}.$

\item \label{item2:LemmaA3} If $\Psi$ satisfies (\ref{zero:cont:ip}), the sublevel sets 
$
T_{a}:=\left\{\mu\in ca_{Z}^{+} : \Psi^{*}(\mu) \leq a\right\}$, $a\in\mathbb{R},
$
are $\sigma(ca_{Z}^{+}, C_{Z})$-compact.
\end{enumerate}
\end{lem}

\begin{proof}
Defining $\varphi(x):=\sup_{y\in\mathbb{R}^{+}}\left\{xy - \Psi(yZ%
\boldsymbol{1})\right\}$, item \ref{item1:LemmaA3} follows as in \cite[Lemma~$\mathbf{A.3%
}$]{Che:Kup:Tan}.
By definition, $\Psi^{*}$ is a $\sigma(ca_{Z}^{+}, C_{Z})$-lower
semicontinuous function. Hence, the sets $T_{a}$ are $\sigma(ca_{Z}^{+},
C_{Z})$-closed. Every $\mu\in T_{a}$ satisfies 
\begin{equation*}
m\left\langle (Z-z)^{+}\boldsymbol{1}, \mu\right\rangle - \Psi\Big(m(Z-z)^{+}%
\boldsymbol{1}\Big) \leq \Psi^{*}(\mu)\leq a ~ \text{for all } m,z\in\mathbb{%
R}^{+}.
\end{equation*}
Hence, by the definition of $\Psi^{*}$ and Assumption~(\ref{zero:cont:ip})
one has that, for every $m\in\mathbb{R}_{+}$, there exists a $z\in\mathbb{R}%
_{+}$ such that 
\begin{equation*}
\left\langle (Z-z)^{+}\boldsymbol{1}, \mu\right\rangle \leq \frac{a+\Psi(0)+1%
}{m}~~\text{for all } \mu\in T_{a}.
\end{equation*}

Let us now consider, for $i\in\{1, \dots, N\}$, the projection map $\pi^{i}:
ca_{Z}^{+} \to ca_{Z}^{+}(\mathbb{R}), \pi^{i}:(\mu^{1}, ..., \mu^{N})
\mapsto \mu^{i}$ 
and the set $\pi^{i}( T_{a}) $. It obviously holds that $\left\langle
(Z-z)^{+}, \mu^{i}\right\rangle \leq \left\langle (Z-z)^{+}\boldsymbol{1},
\mu\right\rangle$, leading to the following inequality (that is Condition~$%
\mathbf{(A.5)}$ in the proof of \cite[Lemma~$\mathbf{A.3}$]{Che:Kup:Tan}): 
\begin{equation*}
\begin{aligned} \lim_{z\to+\infty}\sup_{\mu^{i}\in\pi^{i}(
T_{a})}{\left\langle Z\mathbbm{1}_{\{Z<2z\}}, \mu^{i}\right\rangle} &\leq
\lim_{z\to+\infty}\sup_{\mu\in T_{a}}{\left\langle Z\mathbbm{1}_{\{Z<2z\}},
\mu\right\rangle} \\ &\leq \lim_{z\to+\infty}\sup_{\mu\in
T_{a}}{\left\langle 2(Z-z)^{+}\boldsymbol{1}, \mu\right\rangle} = 0. \\
\end{aligned}
\end{equation*}

From item \ref{item1:LemmaA3}, it follows that $\left\langle Z, \mu^{i}\right\rangle \leq \left\langle Z\boldsymbol{1},
\mu\right\rangle \leq \varphi^{-1}(a) < +\infty$ $\forall\mu\in T_{a}$ (and so $\forall\mu^{i}\in\pi^{i}( T_{a}))$
that is condition $\mathbf{(A.6)}$ in \cite{Che:Kup:Tan} for $\pi^{i}(
T_{a}) $. Identifying $C_{Z}$ with the space of $N$-dimensional vectors of
continuous bounded functions $C_{b}$ by the function $f:X\mapsto\frac{X}{Z}$
and $ca_{Z}^{+}$ with the set of $N$-dimensional vectors of finite Borel
measures $ca^{+}$ by the function $g : \mu\mapsto Zd\mu$, conditions~$%
\mathbf{(A.5)}$ and $\mathbf{(A.6)}$ imply that $\pi^{i}(g(
T_{a}))=h(\pi^{i}( T_{a}))$ is tight, where $h:ca_{Z}^{+}(\mathbb{R})\to
ca^{+}, \mu^{1}\mapsto Zd\mu^{1}$. So one obtains from Prokhorov's theorem
that $h(\pi^{i}( T_{a}))$ is $\sigma\big(ca^{+}(\mathbb{R}), C_{b}(\mathbb{R}%
)\big)$-compact, that is equivalent to $\pi^{i}( T_{a})$ being $\sigma\big(%
ca_{Z}^{+}(\mathbb{R}), C_{Z}(\mathbb{R})\big)$-compact. Clearly, since $%
T_{a}$ is a closed subset of $\pi^{1}( T_{a})\times \dots \times \pi^{N}(
T_{a})$, which is a $\sigma(ca_{Z}^{+}, C_{Z})$-compact set as product of
compact sets, $T_{a}$ is $\sigma(ca_{Z}^{+}, C_{Z})$-compact.
\end{proof}

\begin{thm}[{Multivariate version of {\protect\cite[Theorem~$\mathbf{A.5}$]%
{Che:Kup:Tan}}}]
\label{dual:repr:thm:U} Let $\Psi(X):U_{Z}\to\mathbb{R}$ be an increasing
convex functional satisfying Condition~(\ref{zero:cont:ip}). Then the
following are equivalent:

\begin{enumerate}

\item \label{cond:1:dual:U}$\Psi(X) = \max_{\mu\in
ca_{Z}^{+}}\left\{\left\langle X, \mu\right\rangle - \Psi^{*}(\mu)\right\}$
for all $X\in U_{Z}$

\item \label{cond:2:dual:U} $\Psi(X)\downarrow\Psi(X)$ for all $X\in U_{Z}$
and every sequence $(X_{n})_{n}$ in $C_{Z}$ such that $X_{n}\downarrow X$

\item \label{cond:3:dual:U} $\Psi(X) = \inf_{Y\in C_{Z}, Y\geq X}{\Psi(Y)}$
for all $X\in U_{Z}$

\item \label{cond:4:dual:U} $\Psi^{*}(\mu) = \sup_{X\in
U_{Z}}\left\{\left\langle X, \mu\right\rangle - \Psi(X)\right\}$ for all $%
\mu\in ca_{Z}^{+}$.
\end{enumerate}
\end{thm}

\begin{proof}
The proof follows from Theorem \ref{dual:repr:thm} and Lemma \ref%
{phi*:sublev:comp} as in the univariate case and we thus omit that. This is
because the convergence is intended componentwise and the results can be
applied on each coordinate.
\end{proof}

\bibliographystyle{siam}
\bibliography{Bibliography}

\end{document}